\documentclass[conference,letterpaper]{IEEEtran}

\IEEEoverridecommandlockouts

\usepackage[utf8]{inputenc} 
\usepackage{liustyle}
\usepackage[T1]{fontenc}
\usepackage{url}
\usepackage{ifthen}
\usepackage{cite}

\usepackage{setspace}
\usepackage{svg}

\usepackage{esvect}

\usepackage{newtxtext} 

\begin{document}
\title{The Role of Rank in Mismatched Low-Rank Symmetric Matrix Estimation
\thanks{The first two authors contributed equally to this work, in a random coin order, and ${\dagger}$ marks the corresponding author.
This work was partially supported by National Natural Science Foundation of China Grant No. 12025104.}
}


\author{%
 \IEEEauthorblockN{Panpan Niu, Yuhao Liu, Teng Fu, Jie Fan, Chaowen Deng and Zhongyi Huang$^{\dagger}$}
\IEEEauthorblockA{Department of Mathematical Sciences, Tsinghua University, Beijing, China\\
                   Email:  \{npp21, yh-liu21, fut21, fanj21  and dcw21\}@mails.tsinghua.edu.cn, zhongyih@tsinghua.edu.cn} 
}


\maketitle


\begin{abstract}
    We investigate the performance of a Bayesian statistician tasked with recovering a rank-\(k\) signal matrix \(\bS \bS^{\top} \in \mathbb{R}^{n \times n}\), corrupted by element-wise additive Gaussian noise. This problem lies at the core of numerous applications in machine learning, signal processing, and statistics.  We derive an analytic expression for the asymptotic mean-square error (MSE) of the Bayesian estimator under mismatches in the assumed signal rank, signal power, and signal-to-noise ratio (SNR), considering both sphere and Gaussian signals. Additionally, we conduct a rigorous analysis of how rank mismatch influences the asymptotic MSE. Our primary technical tools include the spectrum of Gaussian orthogonal ensembles (GOE) with low-rank perturbations and asymptotic behavior of \(k\)-dimensional spherical integrals.
\end{abstract}

\section{Introduction}
The problem of low-rank matrix recovery from noisy observations is central to a wide range of applications in machine learning, signal processing, and statistics. Notable examples include community detection under the stochastic block model \cite{moore2017computer,abbe2018community}, matrix completion \cite{keshavan2010matrix,saade2015matrix}, and sparse principal component analysis (PCA) \cite{zou2006sparse,johnstone2009consistency}, among others. In a symmetric setting, the signal matrix $\bS \in \RR^{n \times k}$ is generated with a prior $P_{\bS}$, and the goal is to estimate the matrix $\bS \bS^\top$ from the noisy symmetric observation $\bY$, through element-wise additive Gaussian noise.

Various optimization methods, including rank-penalized \cite{bunea2011optimal}, convex relaxation \cite{candes2010power}, gradient-based \cite{chi2019nonconvex}, and spectral-based \cite{chen2021spectral,candes2011robust,wainwright2019high} algorithms, have been widely applied to low-rank matrix recovery. When structural information about $\bS$ is available, Bayesian methods provide an alternative framework. In particular, the Bayesian minimum mean square error (MMSE) estimator, studied under the "Bayesian-optimal" setting, presumes knowledge of both the prior $P_{\bS}$ and relevant hyperparameters (e.g., SNR). The statistical limitation of the Bayesian MMSE estimator can be characterized by examining the mutual information between observations and signals, and then applying the I-MMSE formula \cite{guo2005mutual} to deduce MMSE and the fundamental information-theoretic constraints. A non-rigorous strategy for computing the mutual information is to view it through the lens of \emph{free energy} and then analyze it using the replica and cavity methods borrowed from statistical physics \cite{korada2009exact,lesieur2015mmse}. These results have been rigorously justified via interpolation techniques \cite{dia2016mutual,lelarge2017fundamental,barbier2019adaptive,barbier2019adaptive2}. Another question we are concerned with is computation limitation: whether the Bayesian MMSE estimator can be computed within polynomial time $O(\text{poly}(n))$. Traditional approaches, such as Markov chain Monte Carlo \cite{salakhutdinov2008bayesian} and belief propagation \cite{pearl2022reverend}, typically require exponential time $O(e^n)$.  To circumvent this computational bottleneck, the approximate message passing (AMP) algorithm \cite{lesieur2017constrained,montanari2021estimation,fan2022approximate, camilli2024decimation} has been proposed for use in high-SNR regimes. Finally, extensions to the non-symmetric low-rank matrix recovery are discussed in \cite{barbier2017layered,kadmon2018statistical,aubin2019spiked,luneau2020high}.

The "Bayes-optimal" framework critically depends on the assumption that the observer possesses exact knowledge of both the prior $P_{\bS}$ and SNR. Under these ideal conditions, the "Nishimori identity" provides a key analytical simplification. In practice, however, such perfect knowledge is rarely available, and one often encounters a \emph{mismatched} setting wherein the observer’s model assumptions deviate from the true underlying signal and noise processes. Recent work on mismatched settings for symmetric rank-one matrix estimation has focused on mean square error (MSE) analyses under scenarios in which the observer knows that both the signal and the noise follow Gaussian distributions, but lacks precise information about their respective powers or the exact SNR \cite{pourkamali2022mismatched,camilli2022inference}. In \cite{pourkamali2022mismatched2}, this analysis is extended to the non-symmetric regime. Furthermore, \cite{barbier2022price} broadens the scope by considering mismatch in both hyperparameters and the noise structure—specifically, allowing the noise to be rotationally invariant while still inferring under a Gaussian noise assumption—and \cite{fu2023mismatched} extends these results to the non-symmetric setting. However, in many practical scenarios, the signal $\bS$ is low-rank but not necessarily rank-one, and its rank may be unknown to the observer. In this work, we analyze the statistical limitations of the Bayesian MMSE estimator under these mismatched conditions, and summarize our main contributions as follows:
\begin{itemize}
    \item We derive the asymptotic MSE of the Bayesian MMSE estimator under mismatches in the signal rank, signal power and SNR, first for spherical signals and then extending the analysis to Gaussian signals. Our primary techniques are based on limitation spectral of GOE with low-rank perturbations \cite{benaych2011eigenvalues} and
    the asymptotic behavior of $k$-dimensional spherical integrals \cite{guionnet2005fourier, collins2007new,guionnet2021asymptotics,maillard2019high}.
    \item We provide a comprehensive theoretical investigation of how rank affects the asymptotic MSE, demonstrating that it is governed by three key factors: \textit{effective rank},  \textit{inference rank}, and \textit{overfitting rank}. 
\end{itemize}

\noindent
\textit{Notation.} 
Let $\langle\cdot,\cdot\rangle$ be the standard inner product in $\RR^n$. Denote by $\delta(\cdot)$ the Dirac delta function, by $\Gamma(\cdot)$ the Gamma function, and by $\SS^{n-1}(r)$ the sphere of radius $r$ in $\RR^n$. Let $U(\Omega)$ denote the uniform distribution over a compact set $\Omega \subset \RR^n$.

\section{Problem Set-up}
We consider the problem of recovering a low-rank signal matrix from noisy observations. Specifically, let $\bY \in \RR^{n\times n }$ be defined by
\begin{equation*}
    \bY = \sqrt{\frac{\lambda_*}{n}} \sum_{i=1}^r \alpha_i \bs_i \bs_i^\top + \bW,
\end{equation*}
where $\lambda_* \in \RR^{+}$ is the true signal-to-noise ratio (SNR). The matrix $\bW \in \RR^{n \times n}$ is a symmetric noise drawn from the \textit{Gaussian orthogonal ensemble} (GOE), in which off-diagonal entries are \iid $\cN(0,1)$  and diagonal entries are \iid $\cN(0,2)$. The vectors $\bs_1, \cdots, \bs_r \in \SS^{n-1}(\sqrt{n})$ represent the true signal components, generated with the following prior
\begin{equation*}
    P(\bS) = \frac{1}{Z_{n,r}} \prod_{i=1}^r \delta \left( \frac{1}{n} \|\bs_i\|_2^2 - 1\right) \prod_{i<j} \delta \left( \frac{1}{n} \langle \bs_i, \bs_j\rangle\right),
\end{equation*}
where $Z_{n,r} = \prod_{i=1}^{r}\left(\frac{2 \pi^{(n-i+1) / 2}}{\Gamma((n-i+1) / 2)} n^{\frac{n-i}{2}}\right)$ is the normalization factor. The scalars $\alpha_1, \cdots, \alpha_r \in \RR^{+}$ denote the respective signal strengths of each component. Our task is to infer $\bS_{\bm{\alpha}} \bS_{\bm{\alpha}}^{\top}$ given the observation $\bY$ where $\bS_{\bm{\alpha}} =  [\sqrt{\alpha_1}\bs_{1}, \cdots, \sqrt{\alpha_r}\bs_{r}]$. In practice, although one knows the uniform-sphere prior and the additive GOE noise structure, the number of signal components (i.e., the rank of $\bS_{\alpha}$), their strengths, and the true SNR are typically unknown. Consequently, we consider inferring the signal matrix by assuming $k$ signal components, each with strength $\beta_i$ for $i=1, \cdots, k$, under an inferred SNR $\lambda$. The mismatched posterior distribution employed for inference is given by 
\begin{equation*}
    P_{\text{mis}}(\mathrm{d} \bX \mid \bY) = \frac{1}{Z_n(\bY)} P(\bX)  e^{-\frac{1}{4} \| \sqrt{\frac{\lambda}{n}}\bX_{\bm{\beta}} \bX_{\bm{\beta}}^{\top} - \bY \|_F^2} \mathrm{d}\bX,
\end{equation*}
where $\bX_{\bm{\beta}} = [\sqrt{\beta_1} \bx_1, \cdots, \sqrt{\beta_k} \bx_k]$ and $Z_n(\bY)$ is the normalization constant (i.e., partition function), defined by
\begin{equation} \label{eq:partition}
    Z_n(\bY) = \int P(\bX) e^{-\frac{1}{4} \| \sqrt{\frac{\lambda}{n}}\bX_{\bm{\beta}} \bX_{\bm{\beta}}^{\top} - \bY \|_F^2} \mathrm{d}\bX.
\end{equation}
The associated \emph{mismatched} Bayes estimator is
\begin{equation*}
    M_{\text{mis}}(\bY) = \int \bX_{\bm{\beta}} \bX_{\bm{\beta}}^{\top} P_{\text{mis}}(\bX \mid \bY) \mathrm{d} \bX.
\end{equation*}
Our goal is to analyze the asymptotic \emph{mismatched} matrix-MSE between $\bS_{\bm{\alpha}} \bS_{\bm{\alpha}}^{\top}$ and $M_{\text{mis}}(\bY)$ given by 
\begin{equation*}
    \mathrm{MSE}_{n}^{\mathrm{Sph}}(\bm{\alpha}, \bm{\beta},\lambda_*, \lambda) = \frac{1}{n^2} \EE_{\bS, \bW} \left[ 
    \| \bS_{\bm{\alpha}} \bS_{\bm{\alpha}}^{\top} - M_{\text{mis}}(\bY)\|_F^2
    \right].
\end{equation*}

\begin{remark}
    The orthogonality requirement on the signal vectors in the prior is largely a technical assumption. In practice, one may sample $\bs_1, \cdots, \bs_r $ \iid from $U(\SS^{n-1}(\sqrt{n}))$. Define $U_{ij} = \cos(\theta(\bs_i, \bs_j))$.  By projective limit theorem \cite{vershynin2018high}, we have 
    \begin{equation*}
        \sqrt{n} U_{ij} \xrightarrow{d} \cN(0,1) ,
    \end{equation*}
    which implies
    \begin{equation*}
        P(\{\exists i,j: n^{-1}|\langle \bs_i, \bs_j\rangle | > \epsilon\}) \leq r^2 e^{-\frac{n \epsilon^2}{2}}
    \end{equation*}
    for sufficiently large $n$. Hence, as $r = O(1)$ and $n \to \infty$, we can effectively treat vectors $\bs_i$ and $\bs_j$ as orthogonal.
\end{remark}


\vspace{-0.2em}
\section{Main result}
\vspace{-0.2em}
This section derives the asymptotic expression for the mismatched matrix-MSE and offers a comprehensive analysis of its structure. We begin by outlining the necessary assumptions and definitions.
\begin{assumption}[Low Rank]\label{ass:low-rank}
\vspace{-0.2em}
     Both the true and inference-model signals have low rank, i.e., $r, k = O(1)\ \text{w.r.t.} \ n$.
     \vspace{-0.2em}
\end{assumption}
\begin{assumption}[Decreasing Power]\label{ass:power}
    In both the true and inference models, the signal powers are in non-increasing order. Specifically, 
    \begin{equation*}
        \alpha_1\geq \alpha_2 \geq \cdots \geq \alpha_r>0, \quad \beta_1 \geq \beta_2 \geq \cdots \geq \beta_k>0.
    \end{equation*}
\end{assumption}
\begin{remark} 
This decreasing power assumption is both natural and entails no loss of generality. Indeed, if the powers are not in non-increasing order, one may simply multiply the model by an appropriate permutation matrix to reorder them, leaving the signal and noise distributions unchanged.
\end{remark}
We now introduce three notions of rank crucial for analysis.
\begin{definition}[Effective Rank]
\vspace{-0.2em}
The \textit{effective rank} \(d\) is defined as the largest integer satisfying
\(\sqrt{\lambda_*} \alpha_d >1\). If no such integer exists, we set $d = 0$.
\vspace{-0.2em}
\end{definition}
\begin{definition}[Inference Rank]
\vspace{-0.2em}
The \textit{inference rank} \(c\) is defined as the largest integer less than $\min \{k, d\}$ satisfying $\sqrt{\lambda_* \lambda} \, \alpha_c \beta_c > 1$.
If no such integer exists, we set $c = 0$.
\vspace{-0.2em}
\end{definition}
\begin{definition}[Overfitting Rank]
\vspace{-0.2em}
If $k>d$, the \textit{overfitting rank} $e$ is defined as the largest integer greater than $d$ satisfying: $\sqrt{\lambda} \beta_e >1$. If no such integer exists, we set $e = d$.
\vspace{-0.2em}
\end{definition}
\vspace{-0.5em}
We now present the asymptotic expression for the mismatched matrix-MSE.\vspace{-0.2em}
\begin{theorem}\label{thm:sph-MSE}
    Suppose Assumptions~\ref{ass:low-rank} and ~\ref{ass:power} hold, and let $(\lambda_*, \lambda) \in \RR^{+} \times \RR^+$. Then the sequence $(\mathrm{MSE}_n^{\mathrm{Sph}}(\bm{\alpha}, \bm{\beta}, \lambda_*, \lambda))_{n \geq 1}$ converges to 
    \begin{equation}\label{eq:sph-MSE}
        \begin{aligned}
            &\lim_{n \to \infty} \mathrm{MSE}_n^{\mathrm{Sph}}(\bm{\alpha}, \bm{\beta}, \lambda_*, \lambda) = \underbrace{\sum_{i=1}^c g^{\mathrm{Sph}}(\alpha_i, \beta_i, \lambda_*, \lambda)}_{\text{inference term}} \\ &\quad \quad \quad 
            + \underbrace{\sum_{i=d+1}^{e} \left( \beta_i - 1 / \sqrt{\lambda}\right)^2}_{\text{overfitting term}} + \underbrace{\sum_{i=1}^{r} \alpha_i^2}_{\text{constant term}},
        \end{aligned}
    \end{equation}  
    where $g^{\mathrm{Sph}}(\alpha, \beta, \lambda, \lambda_*)$ is given by 
    \begin{align*}
    g^{\mathrm{Sph}}(\alpha, \beta, \lambda_*, \lambda) &=\beta\left(\beta-2\alpha\right) +\frac{2\beta}{\alpha}\left(\frac{1}{\lambda_*}-\frac{1}{\sqrt{\lambda_{*}\lambda} }\right)\nonumber\\
    &\quad+2\sqrt{\frac{\lambda_{*}}{\lambda}}\left(\frac{1}{\lambda_{*}}-\frac{1}{\alpha^{2}\lambda_{*}^{2}}\right)+\frac{1}{\alpha^{2}\lambda_{*}\lambda}.
    \end{align*}
    By convention, if in \eqref{eq:sph-MSE} the lower index of a summation exceeds its upper index, that summation is taken to be zero.
\end{theorem}
From Theorem~\ref{thm:sph-MSE}, the asymptotic mismatched matrix-MSE decomposes into three terms that align with the three ranks introduced earlier. 
\begin{itemize}
    \item The effective rank $d$ indicates the number of signal components that are recoverable in an information-theoretic sense; from a random matrix perspective, $d$ corresponds to the number of outlier eigenvalues beyond the semicircle law of the GOE \cite{benaych2011eigenvalues}. If both the true SNR $\lambda_*$ and the signal power $\alpha_i$ are sufficiently small (i.e., $d = 0$), then all signal components are subsumed by noise, no outlier eigenvalues appear, and consequently, no signal components can be recovered.
    \item The inference rank $c$ quantifies how many of these $d$ recoverable components are actually perceived by the inference model. This rank contributes to the \emph{inference term}, the only component in the asymptotic mismatched matrix-MSE that can be negative possibly. When the inference SNR and the inference power $\beta_i$ are sufficiently large, so that $\sqrt{\lambda} \beta_i$  exceeds $1/ (\sqrt{\lambda_*} \alpha_i)$, the inference model perceives the true signal components. If the mismatch between the true parameters $(\lambda_*, \alpha_i)$ and the inferred parameters $(\lambda, \beta_i)$ is small, an extreme case being $\alpha = \beta$ and $\lambda_* = \lambda$, then 
    \begin{equation*}
        g^{\mathrm{Sph}} (\alpha, \beta, \lambda_*, \lambda) = - (\alpha+\frac{1}{\lambda \alpha})^2<0,
    \end{equation*}
    which implies that recovering these components reduces the MSE. However, when the mismatch is large (for example, $\beta= 2 \alpha$ while $\lambda = \lambda_*$), one finds 
    \begin{equation*}
        g^{\mathrm{Sph}} (\alpha, \beta, \lambda_*, \lambda) = \frac{2}{ \lambda} - \frac{1}{\lambda^2 \alpha^2} >0,
    \end{equation*}
    as $\sqrt{\lambda} \alpha>1$, meaning that although the model detects the signal components, the misalignment of parameters inflates the MSE. For indices $i$ ranging from $c+1$ to $d$, the inference model fails to detect these true components (since $\sqrt{\lambda} \beta_i < 1/ (\sqrt{\lambda_*} \alpha_i)$), and they therefore do not affect the asymptotic MSE.
    \item The overfitting rank $e$ specifies the maximum number of signal components that the inference model can perceive. To ensure all effectively recoverable components are inferred, one must set $e$ to be at least $d$. However, setting $e$ too large is detrimental, as the inference model inadvertently perceives pure noise from the true model as signal, contributing the \emph{overfitting term}, which is positive and increases the MSE. For indices \(i\) from \(e+1\) to \(r\), the inference model perceives neither signal nor noise, so these components do not affect the MSE.
    \item The final \emph{constant term} depends solely on the number and power of the true signal components, independent of the inference model.
\end{itemize}
We observe that the relationship between the asymptotic mismatched matrix-MSE and the ranks is not captured merely by comparing $r$ and $k$. Unlike the traditional perspective, where $r>k$ is deemed \emph{under-parameterized} and $r<k$ is deemed \emph{over-parameterized} \cite{bodin2023gradient}, the behavior here is governed by the interplay among the effective rank $d$, the inference rank $c$, and the overfitting rank $e$. Specifically, the under-parameterized regime arises $k<d$, whereas the over-parameterized regime arises when $e>d$. Figure~\ref{fig:cases row} illustrates the composition of the MSE for these two scenarios.
\begin{figure}[tbp] 
    \centering
\includegraphics[width=0.48\textwidth]{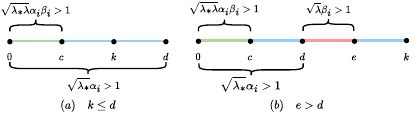}
    \vspace{-0.9em}
    \caption{Under-parameterized  (left) versus over-parameterized  (right). The green region denotes the inference term,  which may decrease the asymptotic mismatched MSE; the red region denotes the overfitting term, which increases it. The blue region ($c<i \leq d$ or $i>e$) contributes nothing to the asymptotic mismatched MSE. The constant term is omitted in the figure.
    }
    \label{fig:cases row}
    \vspace{-1.7em}
\end{figure}

We give a proof sketch of Theorem \ref{thm:sph-MSE} in Section \ref{sec:analysis} and a detailed proof in the Appendix~\ref{app: main them}.

\vspace{-0.5em}
\subsection{MSE with Matched SNR}
We examine how the inference rank and inferred powers $\{\beta_i\}_{i=1}^k$ affect the asymptotic mismatched asymptotic MSE in the SNR-matched setting $\lambda = \lambda_*$. In this regime, 
\begin{align*}
\lim_{n \to \infty} \mathrm{MSE}_n^{\mathrm{Sph}}&= 
\sum_{i=1}^c\left[\beta_i(\beta_i-2\alpha_i) +\frac{2}{\lambda}-\frac{1}{\lambda^2 \alpha_i^{2}}\right]\\
&\quad+\sum_{i=d+1}^e(\beta_i-\frac{1}{\sqrt{\lambda}})^2+\sum_{i=1}^r \alpha_i^2.
\end{align*}
Figure~\ref{fig:MSE_lam_combined} plots the asymptotic mismatched MSE versus SNR $\lambda$  for multiple inference ranks $k$ and power $\beta_i$. In the left subfigure, in the low-SNR regime ($\lambda<1$), the effective rank is $d = 0$; no signal components are detectable, and the MSE equals the total power of the true signal, $3$. In the high-SNR regime, the effective rank jumps to $d=3$. The "Bayes-optimal" setting corresponds to $k=3$ ($c=e=3$),  where all signal components are recovered without introducing any overfitting term; thus, as $\lambda$ increases, the the MSE decreases monotonically and tends to zero (perfect recovery) as $\lambda \to \infty$. The under-parameterized cases correspond to $k=1$ and $k=2$, with inference rank $c= 1, 2$, respectively. Each model recovers its first $c$ signal components, so the MSE decreases with $\lambda$ and converges to the power of the unrecovered components, $2$ for $k=1$ and $1$ for $k=2$. The over-parameterized cases correspond to $k=4$ and $k=5$, both with inference rank $c=3$ and overfitting ranks $e= 4$ and $e=5$, respectively. As $\lambda$ increases, the MSE initially decreases—reflecting the recovery of all true components while the overfitting term remains negligible—but eventually rises as the overfitting term becomes significant, ultimately tending to the power of the extra inferred components: $1$ for $k= 4$ and $2$ for $k= 5$. The right subfigure reveals a counterintuitive phenomenon, consistent with the findings of \cite{pourkamali2022mismatched}: when $\beta_i>\alpha_i$, the asymptotic mismatched MSE initially increases and then decreases with $\lambda$. This behavior arises because the overfitting term—a positive contribution—is introduced when $\sqrt{\lambda}> 1/ \beta_i$, prior to the phase transition where the effective rank jumps to $2$ (occurring when $\sqrt{\lambda}> 0.5$). This phenomenon does not occur when $\beta_i < \alpha_i$, as the critical SNR for introducing the overfitting term exceeds that for the phase transition of the effective rank.

\begin{figure}[tbp] 
    \centering
\includegraphics[width=0.48\textwidth]{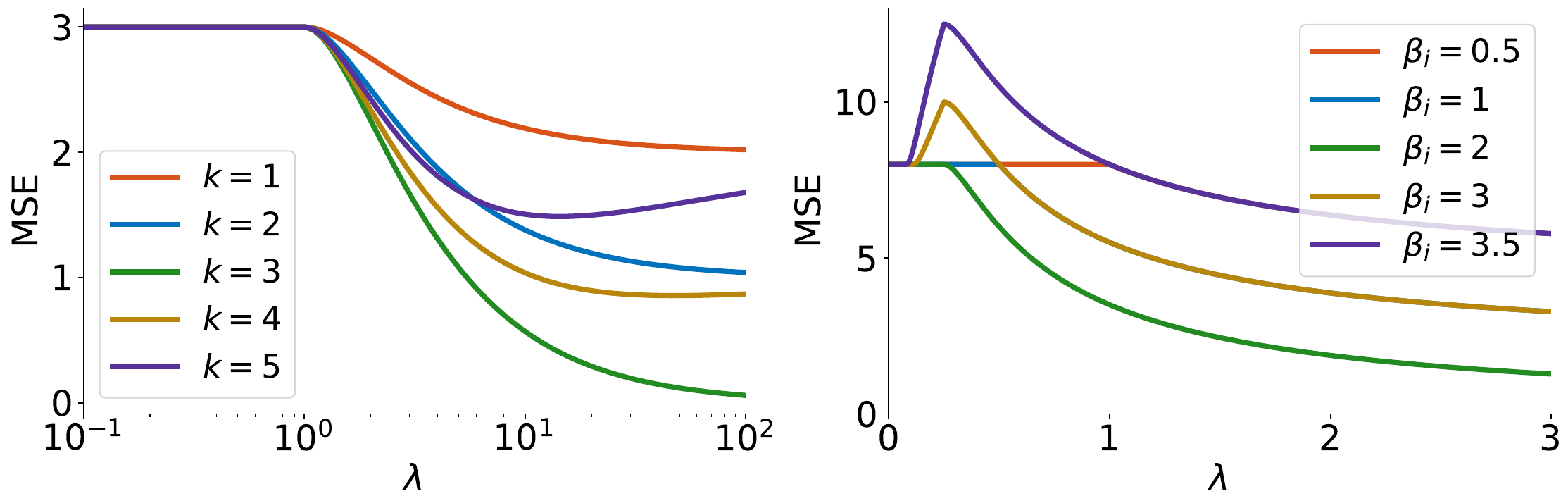}
    \vspace{-1em}
    \caption{ 
    Depiction of asymptotic mismatched MSE as a function of the SNR $\lambda=\lambda_*$ under varying inference-model rank $k$ and signal power $\beta_i$.
    Left: the parameter \(\balpha\) is set to \([1, 1, 1]\), and the inference-model exhibiting various ranks $k$,  with each component power is $1$.
    Right: the parameter \(\balpha\) is set to \([2, 2]\), and the inference-model exhibiting various signal powers with $k=2$.
    }
    \label{fig:MSE_lam_combined}
    \vspace{-1.9em}
\end{figure}

\section{Analysis} \label{sec:analysis}
\vspace{-0.5em}
In this section, we outline the proof of Theorem~\ref{thm:sph-MSE} by first establishing a connection between the mismatched matrix-MSE and the mismatched free energy, and subsequently computing the latter via the asymptotic behavior of $k$-dimensional spherical integrals. 

\vspace{-0.5em}
\subsection{Mismatched free energy}
The \emph{mismatched free energy} is defined as 
\begin{equation}\label{def: fn}
    f_n(\bm{\alpha}, \bm{\beta}, \lambda_*, \lambda) = -\frac{1}{n} \mathbb{E}_{{\bS, \bW}} [\ln Z_n(\bY)],
\end{equation}
where $Z_n(\bY)$ is given in \eqref{eq:partition}. For simplicity, we omit the explicit dependence on $\bm{\alpha}, \bm{\beta}, \lambda_*, \lambda$ but note that both $f_n$ and $\mathrm{MSE}_n$ remain functions of these parameters.  We now present the following lemma, providing a key connection between the mismatched matrix-MSE and the mismatched free energy.
\begin{lemma}[Generalized I-MMSE formula]\label{lem:gen_IMMSE}
   \begin{equation} \label{eq:gen_IMMSE}
       \frac{\partial f_n}{\partial \lambda}  + \left( 2 - \sqrt{\frac{\lambda_*}{\lambda}} \right) \sqrt{\frac{\lambda_*}{\lambda}} \frac{\partial f_n}{\partial \lambda_*}  + \frac{1}{4} \sum_{i=1}^r \alpha_i^2 = \frac{1}{4} \mathrm{MSE}_n.
   \end{equation}
\end{lemma}
\begin{proof}
    The key step is to apply Stein's lemma to $\partial f_n / \partial \lambda$, which enables certain terms in $\partial f_n / \partial \lambda$ and $\partial f_n / \partial \lambda_*$ to cancel. The complete proof can be found in Appendix~\ref{app:proof_gen_IMMSE}.
\end{proof}
\vspace{-0.5em}
\begin{remark}
    This lemma generalizes \cite[Lemma 1]{pourkamali2022mismatched} from the rank-one setting to arbitrary ranks $r$ and $k$. Indeed, setting $r = k = 1$ recovers Lemma~1 of \cite{pourkamali2022mismatched}.
\end{remark}
We note that $\mathrm{MSE}_n$ depends only on the derivatives of $f_n$ with respect to $\lambda$ and $\lambda_*$. Consequently, constant factors in $Z_n(\bY)$  may be discarded without affecting these derivatives. Hence, we rewrite
\begin{equation*}
    Z_n(\bY) = \int  P(\bX)e^{-\frac{1}{4}\frac{\lambda}{n} \sum_{i=1}^k \beta_i^2+\frac{1}{2}\sqrt{\frac{\lambda}{n}}\tr(\bY \bX_{\bbeta}\bX_{\bbeta}^\top)} \mathrm{d} \bX.
\end{equation*}
Because $P(\bX)$ is invariant under orthogonal transformations $\bX \mapsto \boldsymbol{O} \bX$ with $\det(\boldsymbol{O}) = 1$, one can integrate over all such $\boldsymbol{O}$ and average accordingly. This yields
\begin{equation}\label{Haar_partition}
    Z_n(\bY) = e^{-\frac{1}{4} \frac{\lambda}{n}\sum_{i=1}^k \beta_i^2} \mathbb{E}  \int \mathcal{D}  \boldsymbol{O} e^{\frac{\sqrt{\lambda}}{2 \sqrt{n}} \operatorname{Tr} (\bY \boldsymbol{O} \bX_{\bm{\beta}} \bX_{\bm{\beta}}^{\top} \boldsymbol{O}^{\top})},
\end{equation}
where the expectation is with respect to the signal prior in the inference model, and $\mathcal{D} \boldsymbol{O}$ denotes integration over the Haar measure. The principal challenge thus lies in evaluating the Haar measure integral in \eqref{Haar_partition}, which we address in the following subsection.

\vspace{-1em}
\subsection{\texorpdfstring{Asymptotic sphere integrals with rank $k$}{Asymptotic sphere integrals with rank k}}
Evaluating \eqref{Haar_partition} requires a detailed understanding of the following sphere integration 
\begin{equation*}
    I_{n}(\bA, \bB) = \int \mathcal{D} \boldsymbol{O} e^{n \operatorname{Tr} \bA \boldsymbol{O} \bB \boldsymbol{O}^{\top}},
\end{equation*}
where $\bA, \bB \in \mathbb{R}^{n \times n}$ are real symmetric matrices. Such integrals are known in mathematical physics as \emph{Harish-Chandra–Itzykson–Zuber} (HCIZ) integrals and have been widely investigated \cite{harish1957differential}. Since $\mathcal{D} \boldsymbol{O}$ denotes integration over the Haar measure (i.e., the uniform distribution on the all orthogonal matrices, informally speaking), $I_n(\bA, \bB)$ depends only on $\bA$ and the nonzero eigenvalues of $\bB$. Specifically, let $\{\eta_i\}_{i=1}^m$ be the $m$ nonzero eigenvalues of $\bB$, arranged in non-increasing order. Denote by $\text{diag}_n(\{\eta_i\}_{i=1}^m)$ the diagonal matrix in $\RR^{n \times n}$ with these $m$ eigenvalues  on its main diagonal and zeros elsewhere. We obtain
\begin{equation}
    \begin{aligned}
        I_n(\bA, \{\eta_i\}_{i=1}^m) &: = \int \mathcal{D} \boldsymbol{O} e^{n \operatorname{Tr} \bA \boldsymbol{O} \text{diag}_n(\{\eta_i\}_{i=1}^m) \boldsymbol{O}^{\top}}, \\
        &= I_n(\bA, \bB).
    \end{aligned}\vspace{-0.2em}
\end{equation}
Using this notation, we may rewrite $Z_n(\bY)$ as 
\begin{equation*}
    Z_n(\bY) = e^{-\frac{1}{4} \frac{\lambda}{n}\sum_{i=1}^k \beta_i^2}  I_n \left( \bY / \sqrt{n},\{ \sqrt{\lambda} \beta_i / 2\}_{i=1}^k
    \right),
\end{equation*}
where the norm constraint and the orthogonality condition in the prior fix the eigenvalues of $\bX_{\bm{\beta}} \bX_{\bm{\beta}}^{\top}$, so the expectation in \eqref{Haar_partition} can be omitted.
In \cite{guionnet2005fourier}, the asymptotic behavior of sphere integrals is estimated for $m = 1$. Results for $m = O(1)$ with small $\eta_i$ have been established in \cite{guionnet2021asymptotics, collins2007new, potters2020first}, culminating in Theorem~\ref{thm:low_SNR_k_sphere}.
This theorem employs the Stieltjes transform of matrix $\bA$, for $z\in \mathrm{Supp}(\mu_{\bA})^c$, given by $H_{\mu_{\bA}}(z) = \int \frac{1}{z - \lambda} \mathrm{d} \mu_{\bA}(\lambda)$.
Due to space limitations, we put it in appendix~\ref{app_subsec:k_sphe_int_small}. In our setting, the matrix $\bY/ \sqrt{n}$ has a limiting eigenvalue distribution $\mu_{\text{SC}} = \frac{1}{2 \pi} \sqrt{4 - \lambda^2}  \mathrm{d} \lambda$ with $H_{\text{SC}}(z) = \frac{1}{2}(z - \sqrt{z^2 - 4})$~\cite{benaych2011eigenvalues}. Consequently, 
\vspace{-0.5em}
\begin{equation*}
    H_{\min} = -1, \quad H_{\max} = \min \{1, 1/ \sqrt{\lambda_*} \alpha_1 \}.
\end{equation*}
When $\sqrt{\lambda} \beta_1 \le \min \{1, 1/ \sqrt{\lambda_*} \alpha_1 \}$, 
the asymptotic evaluation of $f_n$ decomposes into $k$ rank-one  asymptotic sphere integrals, which can be addressed via the results of \cite{guionnet2005fourier}. However, for larger $\sqrt{\lambda} \beta_i$,  Theorem~\ref{thm:low_SNR_k_sphere} no longer applies. Fortunately, the following proposition from \cite{guionnet2021asymptotics} provides the relevant asymptotic behavior for rank $k$ with any $\sqrt{\lambda} \beta_i$.
\begin{theorem}[Prop. 1 in \cite{guionnet2021asymptotics}] \label{thm:high_SNR_sphere_int}
Suppose that $\bA$ admits a limiting eigenvalue distribution $\mu_A$ with $m$ largest outliers $\gamma_1 \geq \cdots \geq \gamma_m$. Let $\eta_1 \geq \cdots \geq \eta_m \geq 0$. Then 
\vspace{-0.5em}
\begin{equation*}
    \lim_{n \to \infty} \frac{1}{n}\ln I_n(\bA, \{\eta_i\}_{i=1}^m) = \frac{1}{2} \sum_{i=1}^m J(\mu_A, 2 \eta_i, \gamma_i),
\end{equation*}
\vspace{-0.5em}
where $J(\mu, \eta, \gamma) = K(\mu, \eta, \gamma, v(\mu, \eta, \gamma))$ with 
\begin{equation*}
    K(\mu,\eta,\gamma,v)=\eta\gamma+(v-\gamma)H_{\mu}(v)-\int\ln(\eta|v-x|)d\mu(x)-1,
\end{equation*}
and 
\begin{align*}
\vspace{-1.em}
    v(\mu,\eta,\gamma)=\begin{cases}
        \gamma &\text{if}\quad H_{\mu}(\gamma) \leq \eta,\\
        H_{\mu}^{-1}(\eta) &\text{if}\quad H_{\mu}(\gamma) > \eta.\\
    \end{cases}
\end{align*} 
\end{theorem}
Using Theorem~\ref{thm:high_SNR_sphere_int} to compute the asymptotic expression of $f_n$, we set $\eta_i = \sqrt{\lambda} \beta_i / 2$ and $\bA = \bY / \sqrt{n}$. If $k$ is smaller than the number of outliers of $\bY / \sqrt{n}$, then $\eta_i$ may be assigned zero for $k<i \leq m$. Conversely, if $k$ exceeds the number of outliers, the set of outliers can be augmented by introducing $k-m$ additional signal components with power $\alpha_i = (1+\epsilon) / \sqrt{\lambda_*}$ for $m<i \leq k $,  and taking the limit $\epsilon \to 0$ in the final expression.
\vspace{-0.5em}

\vspace{-1.2em}
\subsection{Asymptotic of mismatched free energy and MSE}
\vspace{-0.8em}
Using Theorems~\ref{thm:low_SNR_k_sphere} and \ref{thm:high_SNR_sphere_int}, we now derive the asymptotic expression for the mismatched free energy.
\begin{theorem} \label{thm:free_enery_sph}
     Under the same conditions as Theorem~\ref{thm:sph-MSE}, the asymptotic mismatched free energy is given by
     \vspace{-0.5em}
     \begin{equation}\label{eq:sphere_f}
     \lim_{n \to \infty} f_n =\sum_{i=1}^ch_1(\lambda_{*},\lambda,\alpha_{i},\beta_{i})+\sum_{i=d+1}^eh_2(\lambda,\beta_{i}),
\end{equation}
where 
\vspace{-0.7em}
\begin{align*}
    h_1(\lambda_{*},\lambda,\alpha_{i},\beta_{i})&=\frac{\lambda_{}}{4}\beta_i^2-\frac{\sqrt{\lambda}}{2}\beta_i(\sqrt{\lambda_{*}} \alpha_i+\frac{1}{\sqrt{\lambda_{*}} \alpha_i})\\
    &\quad+\frac{1}{4\lambda_{*}\alpha_i^2}+\frac{1}{2} \ln(\sqrt{\lambda_{*}\lambda}\beta_i\alpha_i)+\frac{1}{2},
\end{align*}
and 
\vspace{-0.7em}
\begin{equation*}
    h_2(\lambda,\beta_{i})=\frac{\lambda}{4}\beta_i^2-\sqrt{\lambda}\beta_i
 +\frac{1}{2} \ln(\sqrt{\lambda}\beta_i)+\frac{3}{4}.
\end{equation*}
\begin{proof}
    See Appendix~\ref{app: fn}.
\end{proof}
\end{theorem}
\vspace{-1.2em}
By substituting the asymptotic mismatched free energy \eqref{eq:sphere_f} into \eqref{eq:gen_IMMSE} and taking derivatives with respect to $\lambda$ and $\lambda_*$, one obtains the asymptotic mismatched matrix-MSE given in \eqref{eq:sph-MSE}.
\vspace{-1.5em}
\subsection{Generalization to Gaussian signal}
\vspace{-0.7em}
In fact, we can compute the asymptotic mismatched matrix-MSE for Gaussian signal (still preserving them orthogonal each other), we have the following theorem 
\begin{theorem} \label{thm:Gau-MSE}
    Suppose $\bS$ is generated with the following prior
    \begin{equation} \label{eq:gau_prior}
        P(\bS) \propto \exp \left \{ -\frac{1}{2} \|\bS\|_F^2\right \} \prod_{i<j} \delta \left( \frac{1}{n} \langle \bs_i, \bs_j\rangle\right),
    \end{equation}
    and with the same conditions in Theorem~\ref{thm:sph-MSE}, then the sequence $(\mathrm{MSE}^{\mathrm{Gau}}_n(\bm{\alpha}, \bm{\beta}, \lambda_*, \lambda))_{n \geq 1}$ converges to 
    \vspace{-1em}
    \begin{equation}\label{eq:gau-MSE}
        \begin{aligned}
            &\lim_{n \to \infty} \mathrm{MSE}^{\mathrm{Gau}}_n(\bm{\alpha}, \bm{\beta}, \lambda_*, \lambda) = \underbrace{\sum_{i=1}^c g^{\mathrm{Gau}}(\alpha_i, \beta_i, \lambda_*, \lambda)}_{\text{inference term}} \\     \vspace{-2em} &\quad \quad \quad
            + \underbrace{\sum_{i=d+1}^e(\frac{1}{\sqrt{\lambda}}-\frac{1}{\lambda\beta_{i}})^2}_{\text{overfitting term}} + \underbrace{\sum_{i=1}^{r} \alpha_i^2}_{\text{constant term}},
        \end{aligned}
        \vspace{-0.5em}
    \end{equation}
    where $g^{\mathrm{Gau}}(\alpha, \beta, \lambda, \lambda_*)$ is given by 
    \vspace{-0.5em}
    \begin{align*}
    g^{\mathrm{Gau}}(\alpha, \beta, \lambda_*, \lambda) &=\alpha^{2}\left(-2\sqrt{\frac{\lambda_*}{\lambda}}+\frac{\lambda_*}{\lambda}\right)+\frac{2}{\sqrt{\lambda_*\lambda}}+\frac{1}{\lambda^{ 2}\beta^2} \vspace{-1em} \\
    &\quad +\frac{2\alpha}{\lambda\beta}\left(1-\sqrt{\frac{\lambda_*}{\lambda}}\right)-\frac{2}{\lambda_*\lambda\beta\alpha}.
    \end{align*}
    \vspace{-1.5em}
\end{theorem}
\begin{proof}
    The proof largely follows the reasoning used for the spherical signal. The main distinction is that the eigenvalues and the Frobenius norm of $\bX_{\bm{\beta}} \bX_{\bm{\beta}}^{\top}$ are not constants. Consequently, one must evaluate $I_n (\frac{\bY}{\sqrt{n}},\{\frac{\sqrt{\lambda} \beta_i}{2 n} \|\bx_i\|_2^2\}_{i=1}^k)$ and apply the Laplace method to derive the asymptotic of $f_n$. 
\end{proof}
\vspace{-0.8em}
\begin{remark}
    In the rank-one case ($r=k=1$), \eqref{eq:gau-MSE} reduces to
    \vspace{-0.5em}
    \begin{equation*}
        \lim_{n \to \infty} \mathrm{MSE}^{\mathrm{Gau}}_n= \begin{cases}
            (\frac{1}{\sqrt{\lambda}}-\frac{1}{\lambda\beta})^2+ \alpha^2 &\text{if cond1},\\
            g^{\mathrm{Gau}}(\alpha, \beta, \lambda_*, \lambda) +\alpha^2 &\text{if cond2},\\
            \alpha^2 &\text{if}\ o.w,
        \end{cases}
        \vspace{-0.3em}
    \end{equation*}
    where cond1 is defined as $\lambda_*\alpha^2\leq1\ \text{and}\ \lambda\beta^2>1$; cond2 is defined as $\lambda_*\alpha^2>1 \ \text{and}\ \lambda_*\lambda\alpha^2\beta^2>1$.
     Under the identifications $\lambda_* = \lambda, \lambda = \lambda^{\prime}, \alpha = \sigma^2, \beta = \sigma^{\prime 2}$, this result recovers the main finding of \cite{pourkamali2022mismatched}. 
\end{remark}
\begin{remark}
    In the matched and high-SNR regime ($\lambda = \lambda_* \to \infty$), the inference term $g^{\mathrm{Gau}}(\alpha, \beta, \lambda, \lambda_*)$ tends to $-\alpha^2$, yielding a strictly negative contribution, while the overfitting term vanishes. This behavior differs from the spherical model, in which one must impose a fixed power $\beta_i$ for the $i$-th inferred signal. By contrast, in the Gaussian setting, the inference model can adapt its estimate of the true signal power and separate signal from noise with matched high SNR. Consequently, under matched high SNR conditions, over-parameterization proves more advantageous.
\end{remark}
\vspace{-0.5em}
\section{Conclusion and further work}
We derived the asymptotic mismatched matrix-MSE of Bayesian MMSE estimators for spherical and Gaussian signals corrupted by GOE noise, under mismatches in signal rank, signal power, and SNR. Our analysis reveals that the effective rank, inference rank, and overfitting rank jointly govern the structure of the asymptotic MSE, with their relative size determining the contributions of the inference term and the overfitting term.
Future directions include studying the convergence rate of the matched matrix-MSE, average-case analysis of mismatched parameters, non-asymptotic behavior of the mismatched matrix-MSE \cite{guionnet2025estimating}, mismatched high-rank inference \cite{pourkamali2024matrix, barbier2024information}, and inference under structured noise \cite{zhang2024spectral, guionnet2025low}. Another key focus is assessing the alignment between our theoretical predictions and algorithmic implementations (e.g., via AMP \cite{rangan2012iterative} or Approximate Survey Propagation \cite{antenucci2019approximate}) in the mismatched setting.

\clearpage
\bibliographystyle{IEEEtran.bst}
\bibliography{ref.bib}

\clearpage
\onecolumn
\begin{appendices}
\section{Generalized I-MMSE formula}\label{app:proof_gen_IMMSE}
We present the proof of Lemma~\ref{lem:gen_IMMSE} (Generalized I-MMSE formula) below.
\begin{proof}[Proof of Lemma~\ref{lem:gen_IMMSE}]
\vspace{-1em}
By substituting
\begin{equation*}
    \bY = \sqrt{\frac{\lambda_*}{n}} \bS_{\balpha} \bS_{\balpha}^\top + \bW
\end{equation*}
into posterior distribution, we obtain
\begin{align*}
        P\{\bX|\bY\} & \propto e^{-\frac{1}{4}\frac{\lambda}{n}\|\bX_{\bbeta}\bX_{\bbeta}^\top\|_F^2+\frac{1}{2}\frac{\sqrt{\lambda_*\lambda}}{n}\tr( \bS_{\balpha} \bS_{\balpha}^\top\bX_{\bbeta}\bX_{\bbeta}^\top)+\frac{1}{2}\sqrt{\frac{\lambda}{n}}\tr(\bW)\bX_{\bbeta}\bX_{\bbeta}^\top} P(\bX)\\
        & \propto e^{-\frac{1}{4}\frac{\lambda}{n}\|\bX_{\bbeta}\bX_{\bbeta}^\top\|_F^2+\frac{1}{2}\frac{\sqrt{\lambda_*\lambda}}{n}\tr( \bS_{\balpha} \bS_{\balpha}^\top\bX_{\bbeta}\bX_{\bbeta}^\top)
        +
        \sqrt{\frac{\lambda}{n}}\sum_{i< j}w_{ij}(\bX_{\bbeta}\bX_{\bbeta}^\top)_{ij}+\frac{1}{2}\sqrt{\frac{\lambda}{n}}\sum_{i=1}^nw_{ii}(\bX_{\bbeta}\bX_{\bbeta}^\top)_{ii}} P(\bX).
\end{align*}
The second proportionality arises due to the symmetry and distribution of the noise matrix $\bW$.
By differentiating mismatched free energy~\eqref{def: fn} with respect to $\lambda$, we obtain
\begin{equation}\label{eq:fn with lambda}
    \frac{\partial}{\partial \lambda} f_n = \frac{1}{4n^2}\EE_{\bS, \bW}\left[\langle \|\bX_{\bbeta}\bX_{\bbeta}^\top\|_F^2 \rangle\right]-\frac{1}{4n^2}\sqrt{\frac{\lambda_*}{\lambda}}\EE_{\bS, \bW}\left[\langle\tr\left(\bS_{\balpha} \bS_{\balpha}^\top\bX_{\bbeta}\bX_{\bbeta}^\top\right)\rangle\right]-\frac{1}{4n\sqrt{n\lambda}} \EE_{\bS, \bW}\left[\langle\tr\left(\bW\bX_{\bbeta}\bX_{\bbeta}^\top\right)\rangle\right],
\end{equation}
where the expectation $\langle f(\bX) \rangle$ is defined by
$\langle f(\bX) \rangle = \int f(\bX) P_{\text{mis}}(\bX \mid \bY) \, \mathrm{d}\bX$,
and $f(\bX)$ is a integrable function of the matrix $\bX$.
The expansion of the last term on the right-hand side leads to the following expression:
\begin{equation*}
    \EE_{\bS, \bW}\left[\langle\tr\left(\bW\bX_{\bbeta}\bX_{\bbeta}^\top\right)\rangle\right]=2\sum_{i< j}\EE_{\bS}\EE_{\bW}\left[\langle w_{ij}(\bX_{\bbeta}\bX_{\bbeta}^\top)_{ij}\rangle\right]+\sum_{i=1}^{n} \EE_{\bS}\EE_{\bW}\left[\langle w_{ii}(\bX_{\bbeta}\bX_{\bbeta}^\top)_{ii}\rangle\right].
\end{equation*}
Next, we apply Stein’s Lemma to handle the off-diagonal and diagonal elements.
For $i\neq j$,
\begin{align*}
    \EE_{\bW}\left[\langle w_{ij}(\bX_{\bbeta}\bX_{\bbeta}^\top)_{ij}\rangle\right]=\EE_{\bW}\left[\frac{\partial\langle w_{ij}(\bX_{\bbeta}\bX_{\bbeta}^\top)_{ij}\rangle}{\partial w_{ij}}\right]=1\cdot\sqrt{\frac{\lambda}{n}}\left(\EE_{\bW}\left[\langle(\bX_{\bbeta}\bX_{\bbeta}^\top)_{ij}^2\rangle\right] -\EE_{\bW}\left[\langle\bX_{\bbeta}\bX_{\bbeta}^\top\rangle^2_{ij}\right]\right).
\end{align*}
And for $i=j$,
\begin{align*}
    \EE_{\bW}\left[\langle w_{ii}(\bX_{\bbeta}\bX_{\bbeta}^\top)_{ii}\rangle\right]=\EE_{\bW}\left[\frac{\partial \langle w_{ii}(\bX_{\bbeta}\bX_{\bbeta}^\top)_{ii}\rangle}{\partial w_{ii}}\right]=2\cdot\frac{1}{2}\sqrt{\frac{\lambda}{n}}\left(\EE_{\bW}\left[\langle(\bX_{\bbeta}\bX_{\bbeta}^\top)_{ii}^2\rangle\right]-\EE_{\bW}\left[\langle\bX_{\bbeta}\bX_{\bbeta}^\top\rangle^2_{ii}\right]\right).
\end{align*}
Substituting the results back into \eqref{eq:fn with lambda}, we obtain
\begin{align*}
    \frac{\partial}{\partial \lambda} f_n = \frac{1}{4n^2}\EE_{\bS,\bW}\left[\|\langle\bX_{\bbeta}\bX_{\bbeta}^\top\rangle\|_F^2\right]-\frac{1}{4n^2}\sqrt{\frac{\lambda_*}{\lambda}}\EE_{\bS, \bW}\left[\langle\tr\left(\bS_{\balpha} \bS_{\balpha}^\top\bX_{\bbeta}\bX_{\bbeta}^\top\right)\rangle\right].\\
\end{align*}
By differentiating mismatched free energy~\eqref{def: fn} with respect to $\lambda_*$, we have
\begin{equation*}
    \frac{\partial}{\partial \lambda_*} f_n = -\frac{1}{4n^2}\sqrt{\frac{\lambda_*}{\lambda}}\EE_{\bS, \bW}\left[\langle\tr\left(\bS_{\balpha} \bS_{\balpha}^\top\bX_{\bbeta}\bX_{\bbeta}^\top\right)\rangle\right].
\end{equation*}
In summary, we obtain
\begin{align*}
    &\frac{\partial}{\partial \lambda} f_n + \left(2 - \sqrt{\frac{\lambda_*}{\lambda}}\right) \sqrt{\frac{\lambda_*}{\lambda}} \frac{\partial}{\partial \lambda_*} f_n + \frac{1}{4}\sum_{i=1}^r \alpha_i^2 \\
    &=\frac{1}{4n^2}\left\{\EE_{\bS,\bW}\left[\|\langle\bX_{\bbeta}\bX_{\bbeta}^\top\rangle\|_F^2\right]-\sqrt{\frac{\lambda_*}{\lambda}}\EE_{\bS, \bW}\left[\langle\tr\left(\bS_{\balpha} \bS_{\balpha}^\top\bX_{\bbeta}\bX_{\bbeta}^\top\right)\rangle\right]\right.\\
    &\qquad\qquad\left.-\bigg(2 - \sqrt{\frac{\lambda_*}{\lambda}}\bigg)\EE_{\bS, \bW}\left[\langle\tr\left(\bS_{\balpha} \bS_{\balpha}^\top\bX_{\bbeta}\bX_{\bbeta}^\top\right)\rangle\right]+\EE_{\bS,\bW} \left[ \|\bS_{\balpha} \bS_{\balpha}^\top\|_F^2\right]\right\}\\
    &=\frac{1}{4n^2}\left\{ \EE_{\bS,\bW}\left[\|\langle\bX_{\bbeta}\bX_{\bbeta}^\top\rangle\|_F^2\right] -2 \EE_{\bS, \bW}\left[\langle\tr\left(\bS_{\balpha} \bS_{\balpha}^\top\bX_{\bbeta}\bX_{\bbeta}^\top\right)\rangle\right]+\EE_{\bS,\bW} \left[ \|\bS_{\balpha} \bS_{\balpha}^\top\|_F^2\right] \right\}\\
    &=\frac{1}{4 n^2} \EE_{\bS, \bW} \left[ 
    \| \bS_{\bm{\alpha}} \bS_{\bm{\alpha}}^{\top} - M_{\text{mis}}(\bY) \|_F^2
    \right]=\frac{1}{4}\mathrm{MSE}_{n}.
\end{align*}
\end{proof}

\section{Computation of the asymptotic mismatched MSE}\label{app: main them}

\subsection{Rank-k sphere integral with small eigenvalues} \label{app_subsec:k_sphe_int_small}
We present the asymptotic behavior of rank-$k$  sphere integrals with small eigenvalues $\{\eta_i\}$ in the following theorem. 
\begin{theorem}[Thm. 6 in \cite{collins2007new}]\label{thm:low_SNR_k_sphere}
    Suppose $\bA$ admits a limiting eigenvalue distribution $\mu_A$ and bound $\{\eta_i\}_{i=1}^m$ as
    \begin{equation*}
        \text{H}_{\min} \leq \eta_m \leq \eta_1 \leq H_{\max},
    \end{equation*}
    then we have 
    \begin{equation*}
        \lim_{n \to \infty}\frac{1}{n} \ln I_{n}(\bA, \{\eta_i\}_{i=1}^m) = \sum_{i=1}^m \lim_{n \to \infty} \frac{1}{n}\ln I_{n}(\bA, \eta_i),
    \end{equation*}
    where 
    \begin{equation*}
        H_{\min}= \lim_{z \uparrow \lambda_{\min(\bA)}} H_{\mu_{\bA}}(z), \quad  H_{\max} = \lim_{z \downarrow \lambda_{\max(\bA)}} H_{\mu_{\bA}}(z).
    \end{equation*}
\end{theorem}
\begin{remark}
    In our setting, $H_{\max} = \min \{1, 1/ \sqrt{\lambda_*} \alpha_1\}$ and $\eta_i = \sqrt{\lambda } \beta_i / 2$. Thus, when $\sqrt{\lambda} \beta_1 > 2\min \{1, 1/ \sqrt{\lambda_*} \alpha_1\}$,  the conditions of the theorem are no longer satisfied.
\end{remark}

\subsection{Calculation of mismatched free energy}\label{app: fn}
We provide the proof of Theorem~\ref{thm:free_enery_sph}, which establishes the asymptotic expression for the mismatched free energy.
\begin{proof}[Proof of Theorem~\ref{thm:free_enery_sph}]
For the matrix \( \frac{\bY}{\sqrt{n}} \), the $d$ largest outliers converge to $\gamma_i=\sqrt{\lambda_{*}}\alpha_i+\frac{1}{\sqrt{\lambda_{*}}\alpha_i},\forall i=1,...,d$, as $n\to \infty$, respectively~\cite{benaych2011eigenvalues}. We proceed by separately considering the cases $k \leq d$ and $k>d$.
\begin{itemize}
    \vspace{1em}
    \item Case 1: $k \leq d$

    \vspace{1em}
    We set $\lambda_i = 0$ for $i = k+1, \cdots, d$, then by Theorem~\ref{thm:high_SNR_sphere_int}, we obtain 
    \begin{equation}\label{case1:free energy}
        \lim_{n\to \infty}f_n=\frac{\lambda}{4}\sum_{i=1}^{k}\beta_i^2 - \frac{1}{2}\sum_{i=1}^k J(\mu_{SC},\sqrt{\lambda}\beta_i,\sqrt{\lambda_{*}}\alpha_i+\frac{1}{\sqrt{\lambda_{*}}\alpha_i}).
    \end{equation}
    For $1\leq i \leq c$, the conditions $\sqrt{\lambda_{*}}\alpha_i\geq1$ and $\sqrt{\lambda_{*}\lambda}\alpha_{i}\beta_{i}\geq 1$ hold.
    The function $J$ is given by:
    \begin{align*}
        J(\mu_{SC},\sqrt{\lambda}\beta_i,\sqrt{\lambda_{*}}\alpha_i+\frac{1}{\sqrt{\lambda_{*}}\alpha_i}) = \sqrt{\lambda}\beta_i(\sqrt{\lambda_{*}} \alpha_i+\frac{1}{\sqrt{\lambda_{*}} \alpha_i})-\frac{1}{2\lambda_{*}\alpha_i^{2}}-\ln(\sqrt{\lambda_{*}\lambda}\beta_i\alpha_i)-1,
    \end{align*}
    with 
    \begin{equation*}
        v(\mu_{SC},\sqrt{\lambda}\beta_i,\sqrt{\lambda_{*}}\alpha_i+\frac{1}{\sqrt{\lambda_{*}}\alpha_i})=\sqrt{\lambda_{*}}\alpha_i+\frac{1}{\sqrt{\lambda_{*}}\alpha_i}.
    \end{equation*}
    For $c+1\leq i \leq k$,  due to $\sqrt{\lambda_{*}}\alpha_i\geq1$ and $\sqrt{\lambda_{*}\lambda}\alpha_{i}\beta_{i}\leq 1$, we deduce that $\sqrt{\lambda}\beta_i\leq1$.
    Therefore,
    \begin{equation*}
        v(\mu_{SC},\sqrt{\lambda}\beta_i,\sqrt{\lambda_{*}}\alpha_i+\frac{1}{\sqrt{\lambda_{*}}\alpha_i})=\sqrt{\lambda}\beta_i+\frac{1}{\sqrt{\lambda}\beta_i}.
    \end{equation*}
    The corresponding $J$ is given by:
    \begin{equation*}
        J(\mu_{SC},\sqrt{\lambda}\beta_i,\sqrt{\lambda_{*}}\alpha_i+\frac{1}{\sqrt{\lambda_{*}}\alpha_i}) = \frac{1}{2}\lambda\beta_i^2.
    \end{equation*}
    Therefore in the case $k\leq d$, we obtain
    \begin{equation*}
        \lim_{n \to \infty} f_n = \sum_{i=1}^c \left[\frac{\lambda}{4} \beta_i^2 - \frac{\sqrt{\lambda}}{2} \beta_i (\sqrt{\lambda_*} \alpha_i + \frac{1}{\sqrt{\lambda_*} \alpha_i}) + \frac{1}{4 \lambda_* \alpha_i^2} + \frac{1}{2} \ln (\sqrt{\lambda_* \lambda } \alpha_i \beta_i) + \frac{1}{2} \right].
    \end{equation*}
    \vspace{1em}
    
    \item Case 2: $k > d$
    
    Due to $k$ exceeds the effective rank $d$, we introduce additional $k-d$ signal components $\bs_{r+1},\cdots,\bs_{r+k-d}$ located in $\SS^{n-1}(\sqrt{n})$, such that orthogonality with \(\bs_1, \ldots, \bs_r\). The signal strengths is specified as $\alpha_i=(1+\epsilon)/\sqrt{\lambda_{*}}, \forall i=r+1,\cdots,r+k-d$. Accordingly, the true model becomes
    \begin{equation*}
        \bY^{\epsilon} = \sqrt{\frac{\lambda_*}{n}} \left( \sum_{i=1}^r \alpha_i \bs_i \bs_i^{\top} + \frac{1+\epsilon}{\sqrt{\lambda_*}} \sum_{i=r+1}^{r+k-d} \bs_i \bs_i^{\top} + \bW\right).
    \end{equation*}
    We choose $\epsilon>0$  sufficiently small so that the inference rank of $\bY$ and $\bY^{\epsilon}$ are identical. Under this perturbation, the number of outliers of $\bY^{\epsilon}$ becomes $k$, with outlier eigenvalues
    \begin{equation*}
        \gamma_i = \sqrt{\lambda_*} \alpha_i + \frac{1}{\sqrt{\lambda_*} \alpha_i}, \; i=1, \cdots, d, \quad \quad \gamma_i = 1+\epsilon + \frac{1}{1+\epsilon}, \; i = d+1, \cdots k.
    \end{equation*}
    As $\epsilon \to 0$, the limiting spectral distributions and outliers of $\bY$ and $\bY^{\epsilon}$ coincide, implying that their asymptotic associated free energies are identical in the limit. Thus, we compute the asymptotic free energy using $\bY^{\epsilon}$ and then take $\epsilon \to 0$, yielding
    \begin{equation*}
        \lim_{\epsilon \to 0} \lim_{n \to \infty} f_n^{\epsilon} = \frac{\lambda}{4} \sum_{i=1}^k \beta_i^2 - \frac{1}{2} \sum_{i=1}^{d} J(\mu_{\text{SC}}, \sqrt{\lambda} \beta_i, \gamma_i) - \frac{1}{2} \sum_{i=d+1}^{k} \lim_{\epsilon \to 0} J(\mu_{\text{SC}}, \sqrt{\lambda} \beta_i, 1+\epsilon + \frac{1}{1 + \epsilon}).
    \end{equation*}
    For $i = d+1, \cdots e$, where $\sqrt{\lambda } \beta_i > H_{\text{SC}}(1+\epsilon + (1+\epsilon)^{-1})$, we find
    \begin{equation*}
        \lim_{\epsilon \to 0} J(\mu_{\text{SC}}, \sqrt{\lambda} \beta_i, 1+\epsilon + \frac{1}{1 + \epsilon}) =  2 \sqrt{\lambda} \beta_i - \ln (\sqrt{\lambda} \beta_i) -\frac{3}{2}.
    \end{equation*}
    For $i > e$, where $\sqrt{\lambda } \beta_i \leq H_{\text{SC}}(1+\epsilon + (1+\epsilon)^{-1})$, we obtain
    \begin{equation*}
        \lim_{\epsilon \to 0} J(\mu_{\text{SC}}, \sqrt{\lambda} \beta_i, 1+\epsilon + \frac{1}{1 + \epsilon}) = \frac{1}{2} \lambda \beta_i^2.
    \end{equation*}
    Thus, the final expression for the asymptotic mismatched free energy is 
    \begin{equation*}
        \lim_{n \to \infty} f_n = \lim_{\epsilon \to 0} \lim_{n \to \infty} f_n^{\epsilon} = \sum_{i=1}^ch_1(\lambda_{*},\lambda,\alpha_{i},\beta_{i})+\sum_{i=d+1}^eh_2(\lambda,\beta_{i}),
    \end{equation*}
    where 
    \begin{equation*}
        h_1(\lambda_{*},\lambda,\alpha_{i},\beta_{i})=\frac{\lambda_{}}{4}\beta_i^2-\frac{\sqrt{\lambda}}{2}\beta_i(\sqrt{\lambda_{*}} \alpha_i+\frac{1}{\sqrt{\lambda_{*}} \alpha_i}) +\frac{1}{4\lambda_{*}\alpha_i^2}+\frac{1}{2} \ln(\sqrt{\lambda_{*}\lambda}\beta_i\alpha_i)+\frac{1}{2},
    \end{equation*}
    and 
    \begin{equation*}
         h_2(\lambda,\beta_{i})=\frac{\lambda}{4}\beta_i^2-\sqrt{\lambda}\beta_i +\frac{1}{2} \ln(\sqrt{\lambda}\beta_i)+\frac{3}{4}.
    \end{equation*}
\end{itemize}
This completes the proof. 
\end{proof}

\subsection{Calculation of asymptotic mismatched MSE}
We now prove the main result, Theorem~\ref{thm:sph-MSE}.
\begin{proof}[Proof of Theorem~\ref{thm:sph-MSE}]
    Given the asymptotic expression for the mismatched free energy in Theorem~\ref{thm:free_enery_sph}, we compute the derivatives:
    \begin{equation*}
        \begin{aligned}
            & \frac{\partial h_1}{\partial \lambda} = \frac{\beta_i^2}{4}-\frac{\beta_i}{4\sqrt{\lambda}}(\sqrt{\lambda_*} \alpha_i+\frac{1}{\sqrt{\lambda_*} \alpha_i})+\frac{1}{4\lambda}, \quad \frac{\partial h_1}{\partial \lambda_*} = -\frac{\beta_i}{4}\sqrt{\frac{\lambda}{\lambda_*}}(\alpha_i-\frac{1}{\lambda_*\alpha_i})-\frac{1}{4\lambda_*^2\alpha_i^{2}}+\frac{1}{4\lambda_*}, \\
            &\frac{\partial h_2}{\partial \lambda} =\frac{1}{4}(\beta_i-\frac{1}{\sqrt{\lambda}})^2, \quad \frac{\partial h_2}{\partial \lambda_*} = 0.
        \end{aligned}
    \end{equation*}
    Substituting these expressions into Lemma~\ref{lem:gen_IMMSE}, we obtain
    \begin{equation*}
        \lim_{n \to \infty} \mathrm{MSE}_n^{\mathrm{Sph}}(\bm{\alpha}, \bm{\beta}, \lambda_*, \lambda,) = \sum_{i=1}^c g^{\text{Sph}}(\alpha_i, \beta_i, \lambda_*, \lambda) + \sum_{i=d+1}^c \left( \beta_i - \frac{1}{\sqrt{\lambda}}\right)^2 + \sum_{i=1}^r \alpha_i^2,
    \end{equation*}
    where 
    \begin{equation*}
        g^{\mathrm{Sph}}(\alpha, \beta, \lambda_*, \lambda) =\beta\left(\beta-2\alpha\right) +\frac{2\beta}{\alpha}\left(\frac{1}{\lambda_*}-\frac{1}{\sqrt{\lambda_{*}\lambda} }\right)\nonumber
    +2\sqrt{\frac{\lambda_{*}}{\lambda}}\left(\frac{1}{\lambda_{*}}-\frac{1}{\alpha^{2}\lambda_{*}^{2}}\right)+\frac{1}{\alpha^{2}\lambda_{*}\lambda}.
    \end{equation*}
    This completes the proof.
\end{proof}
\end{appendices}

\end{document}